\documentclass[ecta,nameyear,draft,nolinenumbers]{econsocart}
\usepackage{graphicx}
\usepackage{color}
\usepackage{setspace}
\usepackage[dvipsnames]{xcolor}
\usepackage{amsmath, amsthm, amssymb, amsfonts, mathtools}
\usepackage{accents}
\usepackage{natbib}
\usepackage{har2nat}
\usepackage{tabularx}
\usepackage{subfig}

\RequirePackage[colorlinks,citecolor=blue,linkcolor=blue,urlcolor=blue,pagebackref]{hyperref}

\startlocaldefs

\theoremstyle{plain}
\newtheorem{theorem}{Theorem}

\newtheorem{lemma}{Lemma}
\newtheorem{claim}{Claim}
\newtheorem{proposition}{Proposition}

\theoremstyle{remark} 

\newtheorem{example}{Example}

\newtheorem{condition}{Condition}

\newcommand{\eps}{\varepsilon}
\newcommand{\ul}{\underline}
\newcommand{\ol}{\overline}
\newcommand{\df}{\mathrm{d}}

\newcommand{\bdis}{\begin{displaymath}}
\newcommand{\edis}{\end{displaymath}}
\newcommand{\beq}{\begin{equation}}
\newcommand{\eeq}{\end{equation}}
\newcommand{\bea}{\begin{eqnarray*}}
\newcommand{\eea}{\end{eqnarray*}}
\newcommand{\bean}{\begin{eqnarray}}
\newcommand{\eean}{\end{eqnarray}}

\newcommand{\R}{\mathbb{R}}

\newcommand{\E}{\mathbb{E}}
\DeclareMathOperator*{\supp}{supp}

\DeclareMathOperator*{\co}{co}

\endlocaldefs

\begin{document}

\begin{frontmatter}

\title{On Monotone Persuasion}
\runtitle{On Monotone Persuasion}

\begin{aug}
\author[id=au1,addressref={add1}]{\fnms{Anton}~\snm{Kolotilin}\ead[label=e1]{akolotilin@gmail.com}}
\author[id=au2,addressref={add1}]{\fnms{Hongyi}~\snm{Li}\ead[label=e2]{hongyi@gmail.com}}
\author[id=au3,addressref={add2}]{\fnms{Andriy}~\snm{Zapechelnyuk}\ead[label=e3]{azapech@gmail.com}}
\address[id=add1]{%
\orgdiv{School of Economics},
\orgname{UNSW Business School}}

\address[id=add2]{%
\orgdiv{School of Economics},
\orgname{University of Edinburgh}}
\end{aug}

\support{Early versions of the results in this paper were presented in working paper versions of \citet{KL} and \citet{KMZ}. We thank Simon Board, Tymofiy Mylovanov, Juuso Toikka, and Alexander Wolitzky for helpful comments. 
Kolotilin gratefully acknowledges financial support from the Australian Research Council Discovery Early Career Research Award DE160100964.
Kolotilin and Li gratefully acknowledge financial support from the Australian Research Council Discovery Project DP240103257.
Zapechelnyuk gratefully acknowledges financial support from the Economic and Social Research Council Grant ES/N01829X/1. 
}
\begin{abstract}
We study monotone persuasion in the linear case, where posterior distributions over states are summarized by their mean. 
We solve the two leading cases where optimal unrestricted signals can be nonmonotone. First, if the objective is s-shaped and the state is discrete, then optimal monotone signals are upper censorship, whereas optimal unrestricted signals may require randomization. Second, if the objective is m-shaped and the state is continuous, then optimal monotone signals are interval disclosure, whereas optimal unrestricted signals may require nonmonotone pooling. We illustrate our results with an application to media censorship.
\end{abstract}

\end{frontmatter}
\newpage
\section{Introduction\label{Intro}}

The literature on Bayesian persuasion has largely focused on the \emph{linear} case, where the state space is one-dimensional and posterior distributions over states are summarized by their mean (e.g., \citealt{GK-RS}, \citealt{KMZL}, \citealt{Kolotilin2017}, and \citealt{DM}). The standard approach has been to analyze \emph{unrestricted} persuasion, where the set of feasible signals is unrestricted. In reality, however, various constraints arise due to incentive, legal, or other practical considerations. Two such constraints are that signals should be \emph{deterministic} and \emph{monotone}. For instance, the bank regulator may be unable to use a stress test that credibly and varifiably randomizes scores or gives higher scores to weaker banks (\citealt{GL}).

These concerns motivate the study of \emph{monotone persuasion} (\citealt{Mensch} and \citealt{OR}) where all feasible signals are deterministic and monotone, so that they partition the state space into convex sets (i.e., intervals and singletons).\footnote{Relatedly, \citet{Ivanov},  \citet{aybas2019persuasion}, \citet{hopenhayn2023optimal},  and \citet{lyu2023coarse} study Bayesian persuasion under a different constraint: the set of signal realizations contains only $N$ elements.} However, these studies do not address the linear case.\footnote{\citet{Mensch} studies general nonlinear monotone persuasion and discusses a special subcase of the linear case in Section 5.2 where the objective function is quadratic, and thus convex or concave, so full or no disclosure is optimal. In contrast, \citet{OR} study specific nonlinear monotone persuasion, which overlaps with the linear case only when the objective is linear, so all signals are optimal.}
In the literature that deals with the linear case, \citet{DM} delineate conditions under which optimal unrestricted persuasion is monotone, so that standard results apply.\footnote{Relatedly, \citet{ABSY} show that, if the state is continuous, it is without loss of optimality to restrict attention to (possibly nonmonotone) deterministic signals or to (possibly stochastic) signals such that a higher state induces a higher lottery over signal realizations with respect to first-order stochastic dominance. More generally, \citet{KZ-WP} show that, regardless of whether the state is discrete or continuous, it is without loss of generality to restrict attention to (possibly stochastic) signals such that a higher state induces a higher lottery over signal realizations with respect to the likelihood ratio order.} It remains an open question what optimal monotone signals look like when optimal unrestricted signals are nonmonotone. We answer this question in the two leading cases that are most relevant in applications.

The first case is the simplest case where randomization is valuable: the state is discrete, and the objective function is s-shaped (convex-concave). Here, it is known that the optimal unrestricted signal is stochastic upper censorship that separates low states, pools high states, and randomizes between separation and pooling at the cutoff state. We show that any optimal monotone signal has the same upper censorship form but does not randomize at the cutoff state.

The second case is the simplest case where nonmonotone pooling of states is valuable: the state is continuous, and the objective function is m-shaped (concave-convex-concave). If optimal unrestricted signals are nonmonotone, then they induce two signal realizations that concavify the objective. We show that any optimal monotone signal in this case is either no disclosure or a cutoff rule that reveals whether the state is below or above a cutoff.

In both cases, existing approaches from the persuasion literature -- such as concavification and linear programming duality -- do not apply because the monotone persuasion problem is not a linear program. Instead, the key step in our approach narrows down the set of possible optimal monotone signals to a simple class by showing that any monotone signal outside of this class is dominated by a  signal in this class.

To illustrate the relevance of our two cases, we use our results to obtain novel economic insights in the media censorship model of \citet{KMZ}, which features a government, heterogeneous citizens, and media outlets. They assume that initially there is a continuum of media outlets and the distribution of citizens' types is unimodal. For any initial set of media outlets, we show that the government's problem of media censorship reduces to a monotone persuasion problem.
Our first case corresponds to the case where there is initially a finite number of media outlets and the distribution of citizens' types is unimodal. In this case, the government permits all sufficiently supportive outlets and censors all other media outlets, which extends the result of \citet{KMZ} on the optimality of upper censorship from the continuous case to the discrete one. Our second case corresponds to the case where there is initially a continuum of media outlets and the distribution of citizens' types is bimodal (i.e., society is polarized). In this case, we obtain a novel result: the government either censors all outlets or permits only one moderate  outlet.

\section{Model}\label{s:model}

A {\it state} $\omega\in [0,1]$ is a random variable with a prior probability distribution function $F$. A {\it signal} reveals information about the state. An {\it objective} $V:[0,1]\mapsto \R$ is a twice continuously differentiable function of the expected state $m$ induced by a signal.   

In many applications, the state is either continuous or discrete. The state is \emph{continuous} if $F$ has a strictly positive density $f$ on $[0,1]$. The state is \emph{discrete} if the support of $F$, denoted by $\supp (F)$, is a finite subset of $[0,1]$. The discrete density is also denoted by $f$.

In an \emph{unrestricted persuasion problem}, a signal can be arbitrarily correlated with the state. 
By Blackwell's informativeness theorem, there exists a signal  that induces a probability distribution $G$ of the expected state $m$ iff the prior distribution $F$ is a mean-preserving spread of $G$  (e.g., \citealt{Kolotilin2017}). Thus, the unrestricted persuasion problem is to find an \emph{optimal unrestricted signal} that maximizes $\int_0^1 V(m)\df G(m)$ over distributions $G$ such that $F$ is a mean-preserving spread of~$G$.

In a \emph{monotone persuasion problem}, a signal is required to be \emph{monotone}: it pools the states into convex sets (i.e., intervals and singletons) and reveals which set contains the realized state. Formally, a monotone signal is an increasing function $\mu:[0,1]\mapsto [0,1]$. W.l.o.g., we identify each signal realization $m$ with the expected state induced by this realization, so $m=\E[\omega|\mu(\omega)=m]$. 
Thus, the monotone persuasion problem is to find an \emph{optimal monotone signal} that maximizes $\int_0^1 V(\mu(\omega))\df F(\omega)$ over monotone signals $\mu$.

\citet{KG} show that full disclosure (resp., no disclosure) is an optimal unrestricted signal, and thus an optimal monotone signal, if the state is discrete and the objective function is convex (resp., concave). \citet{DM} show that an optimal unrestricted signal is monotone if the state is continuous and the objective function is {\it affine closed}. In particular, $V$ is affine closed if it has no m-shaped (concave-convex-concave) region.

We study the two leading cases where an optimal unrestricted signal may be nonmonotone. In Section \ref{s:discrete}, the objective is s-shaped (convex-concave) but the state is discrete. In Section \ref{s:cont}, the state is continuous but the objective is m-shaped (concave-convex-concave).

\section{Discrete State and S-Shaped Objective}\label{s:discrete}

In this section, the state is discrete and the objective is s-shaped.
The objective function $V$ is {\it s-shaped} if there exists $0<\omega_M<1$ such that $V$ is strictly convex on $[0,\omega_M]$ and strictly concave on $[\omega_M,1]$.

A signal is \emph{stochastic upper censorship} if there exist $\omega^*\in\supp(F)$ and $q^*\in[0,1]$ such that states in $[0,\omega^*)$ are separated, states in $(\omega^*,1]$ are pooled, and state $\omega^*$ is separated with probability $q^*$ and pooled with probability $1-q^*$. Let 
\begin{equation*}\label{e:m*}
m^*=\frac{\omega^*(1-q^*)f(\omega^*)+\sum_{\omega>\omega^*}\omega f(\omega)}{(1-q^*)f(\omega^*)+\sum_{\omega>\omega^*} f(\omega)}	
\end{equation*}
be the expected state conditional on the pooling signal realization. A stochastic upper-censorship signal with $(\omega^*,q^*)$ is \emph{deterministic upper censorship} if $q^*\in\{0,1\}$. This is the monotone signal $\mu$ given by
\[
\mu(\omega)=\begin{cases}
\omega, & \omega\in [0,\omega^*),\\
\omega^*, & \omega=\omega^*\text{ and } q^*=1, \\
m^*, & \omega=\omega^*\text{ and } q^*=0, \\
m^*, &  \omega\in(\omega^*,1],
\end{cases}
\]

\citet{AC16} and \citet{KMZ} show that there exist unique $\omega^*\in \supp (F)$ and $q^*\in [0,1]$ satisfying
\begin{equation}\label{e:cens}
V(m^*) +V'(m^*)(\omega^*-m^*) \geq V(\omega^*),\quad \text{with equality if $(\omega^*,q^*)\neq (0,0)$},
\end{equation}
such that the optimal unrestricted signal is stochastic upper censorship with $(\omega^*,q^*)$. Condition \eqref{e:cens} is the first-order necessary condition for optimality, which holds with equality at interior $(\omega^*,q^*)$ and holds with inequality at boundary  $(\omega^*,q^*)= (0,0)$.

\begin{theorem}\label{t:one} If the state is discrete and $V$ is s-shaped, then any optimal monotone signal is upper censorship. Moreover, if $(\omega^*,q^*)$ is given by \eqref{e:cens}, then any optimal monotone signal is upper censorship with $(\omega^*,0)$ or $(\omega^*,1)$.
\end{theorem}

We prove Theorem \ref{t:one} in two steps. The first step shows that any optimal monotone signal is upper censorship. Intuitively, since an s-shaped $V$ is convex for low states (which favours their separation) and concave for high states (which favours their pooling),\footnote{To see that convexity (resp., concavity) of $V$ favours separation (resp., pooling), notice that, in the case with two states $\omega_1<\omega_2$, separation (resp., pooling) yields $V(\omega_1) f(\omega_1) + V(\omega_2)f(\omega_2)$ (resp., $V(\omega_1f(\omega_1) + \omega_2f(\omega_2))$).} it is optimal to separate low states and pool high states, as prescribed by upper censorship. There is a simple linear-programming proof for the case of stochastic upper censorship (see \citealt{KMZ}), but this proof cannot be extended to the case of deterministic upper censorship, because the monotone persuasion problem is a discrete optimization problem when the state is discrete. Instead, for each monotone signal that is not upper censorship, our proof constructs a dominating monotone signal that is  upper censorship. In particular, in the case of three states $\omega_1<\omega_2<\omega_3$, our construction shows that pooling of states $\omega_1$ and $\omega_2$ and separation of state $\omega_3$ is dominated by either full disclosure or no disclosure.

The second step shows that the optimal deterministic upper censorship cutoff coincides with the optimal stochastic upper censorship cutoff. Intuitively, since $V$ is s-shaped, the value of stochastic upper censorship is quasiconcave in $(\omega^*,q^*)$ in the lexicographic order, where $(\omega^*_2,q^*_2)\geq (\omega^*_1,q^*_1)$ iff $\omega^*_2>\omega^*_1$ or $\omega^*_2=\omega^*_1$ and $q^*_2\geq q^*_1$. Thus, if the optimal unrestricted signal is stochastic upper censorship with $(\omega^*,q^*)$, then any optimal monotone signal is deterministic upper censorship with either $(\omega^*,0)$ or $(\omega^*,1)$.

Finally, we remark that the optimal deterministic signal may be nonmonotone, when the state is discrete and the objective is s-shaped.
For example, suppose that state $\omega$ takes only three values $0$, $\varepsilon$, and $1$ with probabilities $(1-q)/2$, $q/2$, and $1/2$ where $\eps,q\in (0,1/2)$. Suppose that $V$ is such that the optimal unrestricted signal is stochastic upper censorship that separates state $0$ with probability $q$ when state $\omega$ takes only two values $0$ and $1$ with equal probabilities. In this case, if $\eps$ is sufficiently small, the optimal deterministic signal pools states $0$ and $1$ and separates state $\varepsilon$.\footnote{Intuitively, as $\varepsilon \to 0$, by continuity, the value of this signal converges to the value of optimal stochastic upper censorship, while the value of any of the other 4 deterministic signals is bounded away from the value of optimal stochastic upper censorship.}

\section{Continuous State and M-Shaped Objective}\label{s:cont}

In this section, the state is continuous and the objective is m-shaped.
The objective function $V$ is {\it m-shaped} if there exist $0<\omega_L<\omega_R<1$ such that $V$ is strictly concave on $[0,\omega_L]$, strictly convex on $[\omega_L,\omega_R]$, and strictly concave on $[\omega_R,1]$.

A monotone signal $\mu$ is {\it interval disclosure} with cutoffs $0\le \omega_L^*\le \omega_R^*\le 1$ if states in the middle interval $[\omega_L^*,\omega_R^*]$ are separated and states in the left interval $[0,\omega_L^*)$ and in the right interval $(\omega_R^*,1]$ are pooled, so
\[
\mu(\omega)=\begin{cases}
m_L^*, &\omega\in[0,\omega_L^*),\\
\omega, & \omega\in[\omega_L^*,\omega_R^*],\\
m_R^*, &  \omega\in(\omega_R^*,1],
\end{cases}
\]
where
\[
m_L^*=\E\big[\omega|\omega\in[0,\omega_L^*]\big] \quad\text{and}\quad m_R^*=\E\big[\omega|\omega\in[\omega_R^*,1]\big]
\]
are the expected states conditional on the pooling signal realizations. 
A monotone signal $\mu$ is a {\it cutoff rule} with cutoff $\omega^*$ if states in the intervals $[0,\omega^*)$ and $(\omega^*,1]$ are pooled. Finally, a monotone signal $\mu$ is {\it no disclosure} if all states in $[0,1]$ are pooled. Note that no disclosure is a special case of a cutoff rule, which is, in turn, a special case of interval disclosure.

It is straightforward to obtain the first-order necessary conditions for optimality, which are similar to conditions in \citet{Kolotilin2017}. If interval disclosure with interior cutoffs $0<\omega^*_L<\omega^*_R<1$ is optimal, then 
\begin{gather}
V(m_L^* )+V'(m_L^*)(\omega_L^*-m_L^*)= V(\omega_L^* ),\label{uml}\\
V(m_R^* )+V'(m_R^*)(\omega_R^*-m_R^*)= V(\omega_R^* )\label{umh}
\end{gather}	
If a cutoff rule with interior cutoff $\omega^*\in (0,1)$ is optimal, then 
\begin{gather}
V(m_L^*)+V'(m_L^*)(\omega^*-m_L^*)=V(m_R^*)+V'(m_R^*)(\omega^*-m_R^*),\label{u*m}
\end{gather}	
where $m_L^*=\mathbb{E}[\omega|\omega\in [0,\omega^*]]$ and $m_R^*=\mathbb{E}[\omega|\omega\in [\omega^*,1]]$.
Also, no disclosure is suboptimal if there exists a cutoff rule with cutoff $\omega^*\in (0,1)$ such that 
\begin{equation}
V(m_L^* )F(\omega^*)	+V(m_R^* )(1-F(\omega^*))> V(\mathbb{E}[\omega]).\label{u*e}	
\end{equation}

\begin{theorem}\label{t-two}
If the state is continuous and $V$ is m-shaped, then any optimal monotone signal is interval disclosure. Moreover:
\begin{enumerate}
	\item If there exist $\omega_L^*,\omega_R^*\in(\omega_L,\omega_R)$ with $\omega_L^*<\omega_R^*$ such that \eqref{uml} and \eqref{umh} hold, then the optimal monotone signal is interval disclosure with cutoffs $\omega_L^*$ and $\omega_R^*$. Also, $m^*_L\in(0,\omega_L)$ and $m^*_R\in(\omega_R,1)$.

	\item Else if there exists $\omega^*\in(0,1)$ such that \eqref{u*e} holds, then an optimal monotone signal is a cutoff rule with some cutoff $\omega^*\in (0,1)$ that satisfies \eqref{u*m} and \eqref{u*e}. Also, $m^*_L\in(0,\omega_L)$ and $m^*_R\in(\omega_R,1)$.

	\item Else, an optimal monotone signal is no disclosure.
\end{enumerate}

\end{theorem}

\begin{figure}[t]
	\begin{tiny} 
	\centering 
	\subfloat[Interval disclosure with cutoffs $\omega_L^*<\omega_R^*$]
	{\label{sub:Switch6}
	\includegraphics[scale=0.58]
	{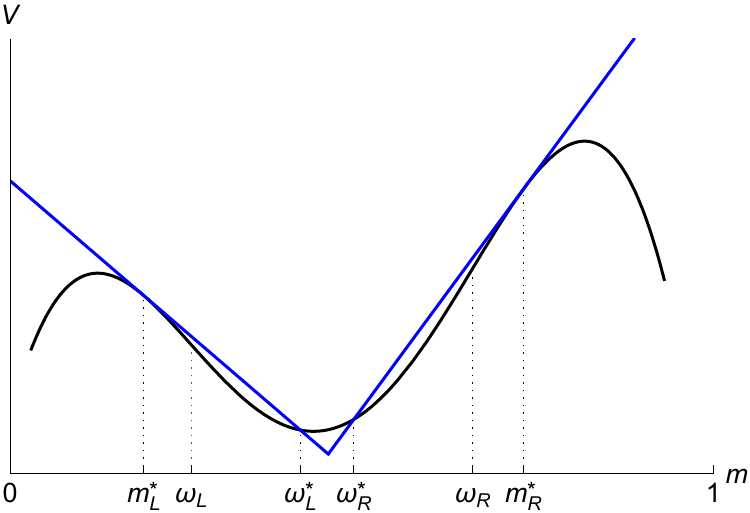}}
	\quad
	\subfloat[Cutoff rule with cutoff $\omega^*$]
	{\label{sub:Switch7}
	\includegraphics[scale=0.58]
	{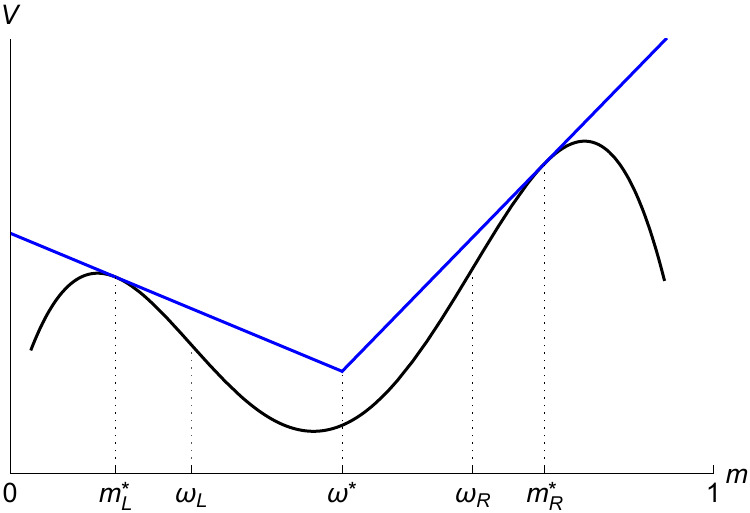}} 
\\
	\subfloat[Cutoff rule with cutoff $\omega^*$]
	{\label{sub:Switch8}
	\includegraphics[scale=0.58]
	{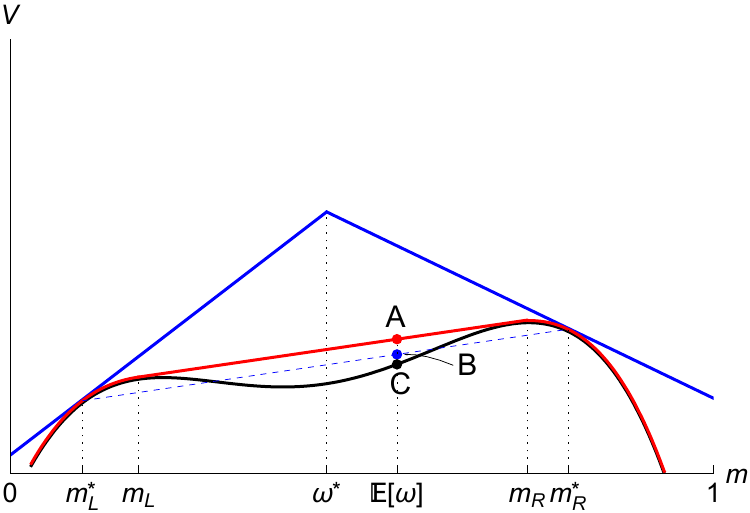}} 
		\quad
	\subfloat[No disclosure]
	{\label{sub:Switch9}
	\includegraphics[scale=0.58]
	{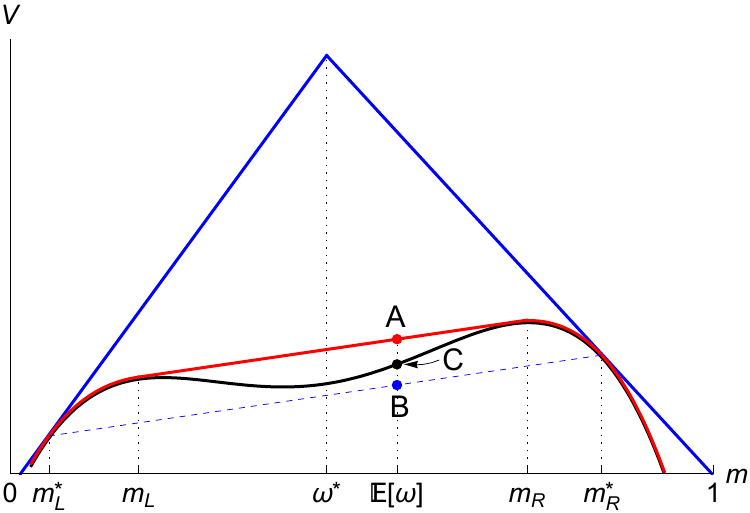}} 
	\end{tiny}
	\caption{Interval disclosure when $V$ is m-shaped.}
	\legend{Moving along Figures \ref{sub:Switch6} $\rightarrow$ \ref{sub:Switch7} $\rightarrow$ \ref{sub:Switch8} $\rightarrow$ \ref{sub:Switch9}, expected states $m^*_L$ and $m^*_R$ move away from each other, which means that the prior distribution $F$ puts increasingly more weight on left and right states (and less weight on middle states). In Figure \ref{sub:Switch6}, the tangents to $V$ at $m^*_L$ and $m^*_R$ cross $V$ at $\omega^*_L$ and $\omega^*_R$. In Figures \ref{sub:Switch7}, \ref{sub:Switch8}, and \ref{sub:Switch9}, the tangents to $V$ at $m^*_L$ and $m^*_R$ intersect at $\omega^*$. In Figures \ref{sub:Switch8} and \ref{sub:Switch9}, points $A$, $B$, and $C$ show the values of bipooling, a cutoff rule, and no disclosure.}
	\label{Fig:PropSwitchustar}
\end{figure}

We prove Theorem \ref{t-two} in two steps. The first step shows that any optimal monotone signal is interval disclosure. Intuitively, since an m-shaped $V$ is convex for middle states (which favours their separation) and concave for extreme states (which favours their pooling), it is optimal to separate middle states and pool extreme states, as prescribed by interval disclosure. There is a simple linear programming characterization of optimal unrestricted signals when $V$ is m-shaped (see \citealt{Kolotilin2017}), but all these signals may be nonmonotone. In this case, the standard approaches from the persuasion literature no longer apply. Instead, for each monotone signal that is not interval disclosure, our proof constructs a dominating monotone signal that is interval disclosure.

The second step delineates conditions under which an optimal monotone signal takes each of the three possible forms of interval disclosure: nondegenerate interval disclosure (Figure \ref{sub:Switch6}), a cutoff rule (Figures \ref{sub:Switch7} and \ref{sub:Switch8}), and no disclosure (Figure \ref{sub:Switch9}). 

If \eqref{uml} and \eqref{umh} hold (Figure \ref{sub:Switch6}), or if \eqref{u*m} and \eqref{u*e} hold and $V'(m_L^*)\le V'(m_R^*)$ (Figure \ref{sub:Switch7}), then the optimal unrestricted signal is interval disclosure (see \citealt{Kolotilin2017}). Otherwise (Figures \ref{sub:Switch8} and  \ref{sub:Switch9}), each optimal unrestricted signal is nonmonotone. This nonmonotone case arises iff the following condition holds (see \citealt{ABSY}).
\begin{condition}\label{c:bipooling}
There is a unique bitangent to $V$ whose tangent points $m_L$ and $m_R$ are such that $0<m_L<\E[\omega]<m_R<1$ and there exists $\omega^{**}\in (m_L,1) $ with $\E[\omega|\omega\in [0,\omega^{**}]]=m_L$ and $\E [\omega|\omega\in [\omega^{**},1]]>m_R$.
\end{condition}
Condition \ref{c:bipooling} says that a cutoff rule with cutoff $\omega^{**}$ induces expected states  $m^{**}_L=m_L$ and $m^{**}_R>m_R$, which intuitively means that the prior distribution $F$ is sufficiently spread out. \citet{KMS} and \citet{ABSY} show that, under Condition \ref{c:bipooling}, each optimal unrestricted signal is \emph{bipooling} in that it induces two expected states $m_L$ and $m_R$ that yield the value $\co V(\mathbb{E}[\omega])$, where $\co V(\mathbb{E}[\omega])$ is the concavification of $V$ at $\mathbb{E}[\omega]$,
\[
\co V(\mathbb{E}[\omega])=V(m_L)\frac{m_R-\E[\omega]}{m_R-m_L}+V(m_R)\frac{\E[\omega]-m_L}{m_R-m_L}.
\]
There exists deterministic bipooling, but then it is necessarily nonmotone. For example, there exist $\omega^{**}_L\in (0,m_L)$ and $\omega^{**}_R\in (m_L,1)$ such that states in $(\omega^{**}_{L},\omega^{**}_R)$ induce expected state $m_L$ and states in $[0,\omega^{**}_L)\cup (\omega^{**}_R,1]$ induce expected state $m_R$. Moreover, there exists stochastically monotone bipooling (i.e., a higher state induces a higher lottery over expected states with respect to first-order stochastic dominance), but then it is necessarily nondeterministic. For example, there exists $q^{**}\in (0,1)$ such that states in $[0,\omega^{**})$ induce expected states $m_L$ and $m_R$ with probabilities $q^{**}$ and $1-q^{**}$, and states in $(\omega^{**},1]$ always induce expected state $m_R$.

We find that, under Condition \ref{c:bipooling}, each optimal monotone signal is either a cutoff rule or no disclosure (Parts 2 and 3 of Theorem \ref{t-two}), and it yields a strictly lower value than the optimal unrestricted signal. In particular, the value of a cutoff rule with cutoff $\omega^*$ is $V(m^*_L)F(\omega^*)+V(m^*_R)(1-F(\omega^*)<\co V(\mathbb{E}[\omega])$,
and the value of no disclosure is $V(\E[\omega])<\co V(\mathbb{E}[\omega])$. If \eqref{u*e} holds, then a cutoff rule dominates no disclosure (Figure \ref{sub:Switch8}), and otherwise no disclosure dominates (Figure \ref{sub:Switch9}).

\section{Application to Media Censorship}\label{Media}

We illustrate our results using the media censorship model of \citet{KMZ}, who characterize optimal media censorship under the following two assumptions. First, there is a continuum of media outlets. Second, the distribution of citizens' types is unimodal. Our results allow to relax these two assumptions, one at a time. Theorem \ref{t:one} yields a characterization of an optimal censorship policy when there is a finite (possibly small) number of media outlets. Theorem \ref{t-two} yields a characterization of an optimal censorship policy when the distribution of citizens' types is bimodal (i.e., society is polarized).

\subsection{Model}

There is a government and a continuum of heterogeneous citizens.  The government's quality $\theta\in [0,1]$ has a distribution $T$ with a strictly positive density on $[0,1]$. Citizens are indexed by $r\in [0,1]$ that has a distribution $V$ with a continuously differentiable density on $[0,1]$. The utility of a citizen of type $r$ is
\[
u(a_r,\theta,r)=(\theta-r)a_r,
\] 
where $a_r\in\{0,1\}$ is the citizen's action.
The government's utility is the aggregate action in the society $\int_0^1 a_r\df V(r)$.

Citizens obtain information about government's quality $\theta$ through media outlets. 
Each media outlet is identified by its editorial policy $c\in [0,1]$, and it endorses action $a=1$ if $\theta\ge c$ and endorses action $a=0$ if $\theta<c$. 
The set of media outlets $C$ is a subset of $[0,1]$.

The government's censorship policy is a set of media outlets $X\subset C$ that are censored. The other media outlets in $C\backslash X$ are permitted to broadcast. The government's problem is to find an \emph{optimal censorship policy} that maximizes its expected utility over censorship policies $X$.

The timing is as follows. First, the government chooses a set $X\subset C$ of censored media outlets. Second,  government's quality $\theta$ is realized, and each permitted media outlet endorses action $a=1$ or $a=0$ according to its editorial policy. Finally, each citizen observes messages from all permitted media outlets, updates beliefs about $\theta$, and chooses an action.

\citet{KMZ} solve the case with $C=[0,1]$ (i.e., there is a continuum of media outlets) and an s-shaped $V$ (i.e., the distribution of citizens' types is unimodal). If it was possible to design any signal about government's quality $\theta$, then deterministic upper censorship with some cutoff $\theta^*\in [0,1)$ would be optimal for the government. Thus, it is optimal to censor all media outlets with editorial policies above $\theta^*$, as this censorship policy implements upper censorship with cutoff $\theta^*$. This approach is valid when $C=[0,1]$ and when the optimal unrestricted signal about $\theta$ is monotone. Our results allow to address more general cases.

\subsection{Reduction to Monotone Persuasion}

We start by showing that the government's problem of media censorship reduces to a monotone persuasion problem with an appropriately defined state.

Consider a censorship policy $X\subset C$. Let $y_X$ be a random variable equal to the conditional expectation of $\theta$ given messages from all media outlets in $C\backslash X$. Let $G_X$ denote the distribution of $y_X$. Each citizen of type $r$ chooses $a_r=1$ iff $r\leq y_X$. Then, the aggregate action is $\int_0^1 a_r\df V(r)=V(y_X)$, and the government's expected utility is $\int_0^1V(y)\df G_X(y)$. Let $\mathcal G_C$ denote the set of distributions $G_X$ induced by all censorship policies $X\subset C$. 

Define the state $\omega$ as the conditional expectation of $\theta$ given messages from all media outlets in $C$. That is, $\omega=y_{\varnothing}$ and its distribution is $F=G_\varnothing$. Consider a monotone signal $\mu$, which is an increasing function satisfying $\E[\omega|\mu(\omega)=m]=m$ for all $m$. Let $G_\mu$ denote the distribution of $m=\mu(\omega)$. Then, the value of $\mu$ is $\int_{0}^1 V(\mu(\omega))\df F(\omega)=\int_{0}^1 V(m)\df G_\mu(m)$. Let $\mathcal G_M$ denote the set of distributions $G_\mu$ induced by all monotone signals $\mu$. 

The next proposition shows that an outcome is implementable by a monotone signal iff it is implementable by a censorship policy. Thus, the government's problem of media censorship reduces to a monotone persuasion problem. 

\begin{proposition}\label{L:M-C}
$\mathcal G_C=\mathcal G_M$.
\end{proposition}

We illustrate the intuition for Proposition \ref{L:M-C} using an example with two media outlets whose editorial policies are $c_1$ and $c_2$.

\begin{example}\label{ex}
 Let $C=\{c_1,c_2\}$ with $0<c_1<c_2<1$. There are three states, $\omega_1=\E[\theta|\theta\le c_1]$, $\omega_2=\E[\theta|c_1\leq \theta\leq c_2]$, and $\omega_3=\E[\theta|\theta\ge c_2]$. There are four censorship policies, (i) $X=\varnothing$, (ii) $X=\{c_1\}$, (iii) $X=\{c_2\}$, and (iv) $X=\{c_1,c_2\}$. They correspond to four monotone signals, (i) full disclosure, (ii) pooling of states $\omega_1$ and $\omega_2$ and separation of state $\omega_3$, (iii) separation of state $\omega_1$ and pooling of states $\omega_2$ and $\omega_3$, and (iv) no disclosure. In particular, no censorship policy implements pooling of states $\omega_1$ and $\omega_3$ and separation of state $\omega_2$.	
\end{example}

\subsection{Discrete Unimodal Case}

Suppose that there is a finite number of media outlets (i.e., state $\omega$ is discrete), and that the distribution of citizens' types is unimodal (i.e., $V$ is s-shaped). By Theorem \ref{t:one}, the government optimally censors all media outlets whose editorial policies are above some cutoff.
That is, all censored media outlets are less supportive than all permitted media outlets in that they endorse the government's preferred action less frequently. To provide the intuition, we discuss  the effect of censoring the least supportive media outlet when there are only two media outlets.
\addtocounter{example}{-1}
\begin{example}[continued]
Relative to free media, by censoring media outlet $c_2$, the government gains support of moderate types in $(\omega_2,m^*)$, where $m^*=\E[\theta| \theta\geq c_1]$, when government's quality $\theta$ is between $c_1$ and $c_2$, but it loses support of opponent types in $(m^*,\omega_3)$ when government's quality $\theta$ is above $c_2$. The gain exceeds the loss, because there are fewer opponent types than moderate types, when the distribution of citizens' types is unimodal.	
\end{example}

The government may gain from using more sophisticated tools of media control than media censorship. For example, the government may prefer to replace media outlets with one government's media outlet that aggregates information from media outlets, possibly adding random noise. We now show that optimal media control may take intricate forms even when there are only two media outlets.
\addtocounter{example}{-1}
\begin{example}[continued] The optimal unrestricted signal may take the form of stochastic upper censorship where state $\omega_2$ is separated with probability $q\in (0,1)$. This signal is implemented by letting media outlet $c_1$ to broadcast freely and by randomly influencing media outlet $c_2$ as follows. With probability $q$, media outlet $c_2$ broadcasts freely. With probability $1-q$, media outlet $c_2$ is forced to repeat the message of media outlet $c_1$. In turn, the optimal deterministic signal may pool states $\omega_1$ and $\omega_3$ and separate state $\omega_2$. This signal is implemented by letting citizens observe only whether media outlets $c_1$ and $c_2$ send the same message or not.
\end{example}

\subsection{Continuous Bimodal Case}

Suppose that there is a continuum of media outlets $C=[0,1]$ (i.e., state $\omega$ is continuous with distribution $F=T$), and that the distribution of citizens' types is bimodal (i.e., society is polarized). For illustration, we restrict attention to the case where $V$ is m-shaped and the distribution of the government's quality is sufficiently spread out in that Condition \ref{c:bipooling} holds.

The government either optimally censors all media outlets (Part 3 of Theorem \ref{t-two}) or permits only one media outlet with a moderate editorial policy $c^*\in (0,1)$ (Part 2 of Theorem \ref{t-two}). It may seem counterintuitive that the government optimally censors not only the least supportive media outlets (as in the unimodal case) but also the most supportive ones. Intuitively, the bimodal case corresponds to a polarized society where most citizens are either supporters or opponents, rather than moderates. Thus, censoring most supportive media outlets ensures that supporters continue to choose the government's preferred action even if no permitted media outlets endorse it.

We now discuss two forms of optimal media control, which outperform optimal media censorship. Let $0<m_L<m_R<1$, $0<\omega^{**}_L<\omega^{**}<\omega^{**}_R<1$, and $q^{**}\in (0,1)$ be as in Section \ref{s:cont}. The first form is deterministic but nonmotone. Citizens observe only whether media outlets $c_L=\omega^{**}_L$ and $c_R=\omega^{**}_R$ send the same message or not. The second form is stochastically monotone but nondeterministic. Citizens observe the message of only one media outlet $c=\omega^{**}$ that is randomly influenced as follows. With probability $q^{**}$, media outlet $c$ broadcasts freely. With probability $1-q^{**}$, media outlet $c$ is forced to endorse the government's preferred action, regardless of the government's quality.

\section{Conclusion}

A lot of progress has been made on optimal persuasion when the objective is posterior mean measurable.
But little is known about optimal monotone persuasion, beyond when optimal persuasion turns out to be monotone. Optimal persuasion can be nonmonotone when the state is discrete, which requires randomization, or when the objective function is irregular, which requires nonmonotone pooling. 
We provide two theorems that characterize optimal monotone persuasion in most prominent such cases, so they can be used as off-the-shelf results in follow-up work. Our proofs identify a candidate class of optimal monotone signals and show that any monotone signal outside of this class is dominated by a signal in this class. This approach can be applied more generally by suitably adjusting a candidate class of optimal monotone signals.

We make a case for monotone persuasion in the context of media censorship. But there are many other considerations that  lead to monotone persuasion. For example, a nonmonotone grading policy that gives better grades to worse performing students may be viewed as unfair or illegitimate and may be manipulated by strategic students. Moreover, the monotonicity restriction may arise due to Mirrlees incentive constraints (e.g., \citealt{Rayo} and \citealt{KL}). Finally, monotone persuasion is equivalent to deterministic delegation (see \citealt{KZ-WP}), so our results are also relevant for the delegation literature, which has primarily focused on deterministic mechanisms.

\newpage

\bibliographystyle{ecta-fullname} %
\bibliography{persuasionlit}  %

\newpage

\begin{appendix}

\section{Proofs}

\subsection{Proof of Theorem \ref{t:one}}

Let $\supp(F)=\{\omega_1,\dots,\omega_n\}$, with natural $n\geq 2$ and $\omega_1<\ldots<\omega_n$. For each $1\leq i<j\leq n$, denote $f_j=f(\omega_j)$, $f_{i:j}= f_i+\ldots+f_j$, and $m_{i:j}=(\omega_i f_i+\ldots+\omega_jf_j)/f_{i:j} $.

\begin{lemma}
	Each optimal monotone signal is deterministic upper censorship.
\end{lemma}
\begin{proof}
Suppose by contradiction that there exists an optimal monotone signal $\mu$ that is not deterministic upper censorship. Then there exist $1\leq i<j<k\leq n$ and two signal realizations: $m_{i:j}$ that pools states $\{\omega_i,\dots,\omega_j\}$ and $m_{j+1:k}$ that pools states $\{\omega_{j+1},\dots,\omega_k\}$. Let $\mu_-$ and $\mu_+$ be monotone signals that differ from $\mu$ only in that $\mu_-$ merges signal realizations $m_{i:j}$ and $m_{j+1:k}$ of $\mu$ into one signal realization $m_{i:k}$ and $\mu_+$ splits signal realization $m_{i:j}$ of $\mu$ into two signal realizations: $m_{i:j-1}$ and $\omega_j$. Denote the value of signals $\mu_-$, $\mu$, and $\mu_+$ by $W_-$, $W$, and~$W_+$. 

To obtain a contradiction, it suffices to show that $W\geq W_+$ implies $W< W_-$. So, suppose  that $W\geq W_+$, which is equivalent to
\beq
V(m_{i:j}) - \frac{\omega_j-m_{i:j}}{\omega_j-m_{i:j-1}}V(m_{i:j-1})-\frac{m_{i:j}-m_{i:j-1}}{\omega_j-m_{i:j-1}}V(\omega_j)\ge 0.\label{e:W_+}
\eeq
Since $V$ is strictly convex on $[0,\omega_M]$ and \eqref{e:W_+} holds, it follows that $\omega_M< \omega_j$.

We now show that $W< W_-$, which is equivalent to
\beq
V(m_{i:k})-\frac{m_{j+1:k}-{m_{i:k}}}{m_{j+1:k}-{m_{i:j}}}V(m_{i:j})-\frac{m_{i:k}-{m_{i:j}}}{m_{j+1:k}-{m_{i:j}}}V(m_{j+1:k})>0.\label{e:W_-}
\eeq
If $\omega_M\leq m_{i:j}$, then \eqref{e:W_-} follows from strict concavity of $V$ on $[\omega_M,1]$. So, suppose that $\omega_M\in (m_{i:j},\omega_j)$. In summary, we have
\beq\label{e:lllog}
m_{i:j-1}<m_{i:j}<\omega_M<\omega_j<m_{j+1:k}, \quad\text{and}\quad m_{i:j}<m_{i:k}<m_{j+1:k}.
\eeq
Since $V$ is strictly convex on $[0,\omega_M]$ and strictly concave on $[\omega_M,1]$, by \eqref{e:lllog}, we have
\begin{align}
&\frac{\omega_M-m_{i:j}}{\omega_M-m_{i:j-1}}V(m_{i:j-1})+\frac{m_{i:j}-m_{i:j-1}}{\omega_M-m_{i:j-1}}V(\omega_M)-V(m_{i:j})>0,\label{e:W0}\\
&V(\omega_j)-\frac{m_{j+1:k}-\omega_j}{m_{j+1:k}-\omega_M}V(\omega_M)-\frac{\omega_j-\omega_M}{m_{j+1:k}-\omega_M}V(m_{j+1:k})> 0.\label{e:W1}\\
&\frac{m_{i:k}-m_{i:j}}{m_{i:k}-m_{i:j-1}}V(m_{i:j-1})+\frac{m_{i:j}-m_{i:j-1}}{m_{i:k}-m_{i:j-1}}V(m_{i:k})-V(m_{i:j})> 0, \quad\text{if $m_{i:k}\le \omega_M$},\label{e:W2}\\
&V(m_{i:k})-\frac{m_{j+1:k}-m_{i:k}}{m_{j+1:k}-\omega_M}V(\omega_M)-\frac{m_{i:k}-\omega_M}{m_{j+1:k}-\omega_M}V(m_{j+1:k})> 0, \quad\text{if $m_{i:k}\ge \omega_M$}.\label{e:W3}
\end{align}
If $m_{i:k}\le \omega_M$, then adding the inequalities \eqref{e:W_+}, \eqref{e:W0}, \eqref{e:W1}, and \eqref{e:W2} multiplied by 
$(m_{j+1:k}-\omega_M)(m_{i:k}-m_{i:j})(\omega_j-m_{i:j-1})$, 
$(m_{j+1:k}-\omega_j)(m_{i:k}-m_{i:j})(\omega_M-m_{i:j-1})$, 
$(m_{i:j}-m_{i:j-1})(m_{i:k}-m_{i:j})(m_{j+1:k}-\omega_M)$, and 
$(\omega_j-\omega_M)(m_{j+1:k}-m_{i:j})(m_{i:k}-m_{i:j-1})$, respectively, yields \eqref{e:W_-}. 
If $m_{i:k}\ge \omega_M$, then adding the inequalities \eqref{e:W_+}, \eqref{e:W0}, \eqref{e:W1}, and \eqref{e:W3} multiplied by 
$(m_{j+1:k}-m_{i:k})(\omega_j-m_{i:j-1})(\omega_M-m_{i:j})$, 
$(m_{j+1:k}-m_{i:k})(\omega_M-m_{i:j-1})(\omega_j-m_{i:j})$, 
$(m_{j+1:k}-m_{i:k})(m_{i:j}-m_{i:j-1})(\omega_M-m_{i:j})$, and 
$(m_{j+1:k}-m_{i:j})(m_{i:j}-m_{i:j-1})(\omega_j-\omega_M)$, respectively, yields \eqref{e:W_-}.

Intuitively, in either case ($m_{i:k}\leq \omega_M$ or $m_{i:k}\geq \omega_M$), we have four linear inequalities (\eqref{e:W_+}, \eqref{e:W0}, \eqref{e:W1}, and \eqref{e:W2} or \eqref{e:W3}) with 6 variables ($V(m_{i:j-1})$, $V(m_{i:j})$, $V(\omega_M)$, $V(\omega_j)$, $V(m_{j+1:k})$, and $V(m_{i:k})$). Fourier-Motzkin elimination of three variables ($V(m_{i:j-1})$, $V(\omega_M)$, and $V(\omega_j)$) yields one inequality \eqref{e:W_-} with 3 variables ($V(m_{i:j})$, $V(m_{j+1:k})$, and $V(m_{i:k})$). Figure \ref{F:discrb} illustrates why \eqref{e:W_-} holds.
\end{proof}

\begin{figure}[ht]
\includegraphics[width=300pt]{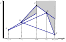}
\caption{Upper censorship when $V$ is s-shaped}\label{F:discrb}
\legend{Point $B$ is above line $AD$, because $W\geq W_+$. Point $C$ is above line $AB$, because $V$ is convex on $[0,\omega_M]$. Point $E$ is below line $CD$, because $V$ is concave on $[\omega_M,1]$. Point $(m_{i:k}, V(m_{i:k}))$ is in the shaded area, because $V$ is convex on $[0,\omega_M]$ and concave on $[\omega_M,1]$. Condition \eqref{e:W_-} states that the shaded area is above line $BE$.}
\end{figure}

We now show that the optimal deterministic upper censorship cutoff coincides with the optimal stochastic upper censorship cutoff. For each $z\in [\omega_1,\omega_n]$, define
\begin{gather*}
j(z)=\max\{i\in\{1,\ldots,n\}:\omega_i\le z\},\\ 
q(z)=\frac{z-\omega_{j(z)}}{\omega_{j(z)+1}-\omega_{j(z)}},\\
m(z)=\frac{(1-q(z))f_{j(z)}\omega_{j(z)}+\sum\nolimits_{i>j(z)}f_{i}\omega_i}{(1-q(z))f_{j(z)}+\sum\nolimits_{i>j(z)}f_{i}},\\
W(z)=\sum_{i<j(z)} f_iV(\omega_i)+q(z)f_{j(z)}V(\omega_{j(z)})+\left((1-q(z))f_{j(z)}+\sum\nolimits_{i>j(z)}f_{i}\right)V(m(z)).
\end{gather*}
Thus, every $z\in [\omega_1,\omega_n]$ represents a stochastic upper censorship signal with $(\omega_{j(z)},q(z))$, where $m(z)$ is the expected state conditional on the pooling signal realization and $W(z)$ is the value of this signal. Conversely, every stochastic upper censorship signal with $(\omega_j,q)$ in $\{\omega_1,...,\omega_{n-1}\}\times [0,1]$ can be represented by $z=(1-q)\omega_{j}+q\omega_{j+1}\in[\omega_1,\omega_n]$.\footnote{An upper censorship with $(\omega_n,q)$, for any $q\in [0,1]$, is the same as the upper censorship with $(\omega_{n-1},1)$.}  Also note that $z$ represents deterministic upper censorship iff $z\in\{\omega_1,...,\omega_n\}$.

Observe that $m(z)$ and $W(z)$ are continuous by construction. Letting
\[
\Delta (\omega,m)=V(\omega)-V(m)-V'(m)(\omega-m)
\]
and taking the derivative of $W(z)$ at $z\notin \{\omega_1,\ldots,\omega_n\}$, we obtain
\begin{align*}
m'(z)&=\frac{q'(z)\sum\nolimits_{i>j(z)} f_i(\omega_i-\omega_{j(z)})}{\left((1-q(z))f_{j(z)}+\sum\nolimits_{i>j(z)}f_{i}\right)^2}=\frac{q'(z)f_{j(z)}(m(z)-\omega_{j(z)})}{(1-q(z))f_{j(z)}+ \sum\nolimits_{i>j(z)}f_{i}},\\
W'(z)&=q'(z)f_{j(z)}(V(\omega_{j(z)})-V(m(z))+m'(z)\Big((1-q(z))f_{j(z)}+\sum_{i>j(z)}f_{i}\Big)V'(m(z))\\
&=q'(z)f_{j(z)}\big(V(\omega_{j(z)})-V(m(z))-V'(m(z))(\omega_{j(z)}-m(z))\big)\\
&=\frac{f_{j(z)}}{\omega_{j(z)+1}-\omega_{j(z)}}\Delta(\omega_{j(z)},m(z)).
\end{align*}

\begin{claim}\label{c:sc}
	For all $\omega_1\le z<z'<\omega_n$, we have
	\[
	\Delta(\omega_{j(z)},m(z))\le 0\implies 	\Delta(\omega_{j(z')},m(z'))<0.
	\]
\end{claim}
\begin{proof}
Suppose by contradiction that $\omega_1\leq z<z'<\omega_n$, $\Delta(\omega_{j(z)},m(z))\leq 0$, and $\Delta(\omega_{j(z')},m(z'))\geq 0$.
By definition, $\omega_{j(z)}$ is increasing in $z$,  $m(z)$ is strictly increasing in $z$, and $\omega(z)<m(z)$, for $z\in [\omega_1,\omega_n)$. Thus, letting $\omega=\omega_{j(z)}$, $\omega'=\omega_{j(z')}$, $m=m(z)$, and $m'=m(z')$, we have $\omega\leq \omega'$, $m<m'$, $\omega<m$, and $\omega'<m'$.  
By integration by parts,
\begin{align*}
\int_{\omega}^{m} V''(x) (x-\omega)\df x &=V'(x)(x-\omega)\big|_{\omega}^{m}-\int_{\omega}^{m}V'(x)\df x	\\
&=V'(m) (m-\omega)-(V(m)-V(\omega))=\Delta(\omega,m). 
\end{align*}
Since $V$ is strictly convex on $[0,\omega_M]$ and strictly concave on $[\omega_M,1]$, we have $V''(x)>0$ for almost all $x\in [0,\omega_M]$ and  $V''(x)<0$ for almost all $x\in [\omega_M,1]$. So, since $\omega<m$ and $\Delta(\omega,m)\leq 0$, we have $m>\omega_M$. Similarly, since $\omega'<m'$ and $\Delta(\omega',m')\geq 0$, we have $\omega'<\omega_M$. Then we obtain a contradiction
\begin{gather*}
\Delta(\omega',m')=\int_{\omega'}^{m'} V''(x) (x-\omega')\df x < \int_{\omega'}^{m} V''(x) (x-\omega')\df x\\
=\int_{\omega'}^{m} V''(x) (x-\omega)\frac{x-\omega'}{x-\omega}\df x \notag
\leq \frac{\omega_M-\omega'}{\omega_M-\omega} \int_{\omega'}^{m} V''(x) (x-\omega)\df x\\
\leq\frac{\omega_M-\omega'}{\omega_M-\omega} \int_{\omega}^{m}V''(x)(x-\omega)\df x = \frac{\omega_M-\omega'}{\omega_M-\omega}\Delta(\omega,m)\leq 0,
\end{gather*}
where the first inequality holds because $\omega_M<m< m'\leq 1$ and $V$ is strictly concave on $[\omega_M,1]$, the second inequality holds because $V$ is convex on $[0,\omega_M]$, concave on $[\omega_M,1]$, and $(x-\omega')/(x-\omega)$ is increasing in $x$, the third inequality holds because $0\leq \omega\leq \omega'<\omega_M$ and $V$ is convex on $[0,\omega_M]$, and the fourth inequality holds because $\Delta (\omega,m)\leq 0$.
\end{proof}

Since $f_{j(z)}/(\omega_{j(z)+1}-\omega_{j(z)})>0$, Claim \ref{c:sc} implies that $W'(z)$ is strictly single crossing from above on $[\omega_1,\omega_n)$. This implies that the optimal unrestricted signal is unique and is stochastic upper censorship with some cutoff $\omega^*$. Furthermore, this implies that each optimal monotone signal is deterministic upper censorship with the same cutoff $\omega^*$ and some $q^{**}\in \{0,1\}$.
\qed

\subsection{Proof of Theorem \ref{t-two}}

It is convenient to represent a monotone signal by a {\it pooling set} $P\subset [0,1]$ of states that are not separated by this signal. Since the state is continuous, w.l.o.g., each pooling interval is open. Thus, the pooling set is a union of some disjoint nonempty open intervals, $P=\bigcup_i(\ul\omega_i,\ol\omega_i)$.\footnote{We define open sets in $[0,1]$ rather than in $\mathbb{R}$; e.g., $[0,1/2)\cup(1/2,1]$ is open.} A pooling set $P$ corresponds to the monotone signal $\mu_P$ given by
\[
\mu_P(\omega)
=
\begin{cases}
\omega, &\text{$\omega\notin (\ul\omega_i,\ol\omega_i)$ for all $i$},\\
m_i, &\text{$\omega \in (\ul\omega_i,\ol\omega_i)$ for some $i$},
\end{cases}
\]
where $m_i=\E [\omega|\omega\in [\ul \omega_i,\ol \omega_i]]$.
The distribution $G_P$ of $\mu_P(\omega)$ is given by
\begin{equation*}\label{GP}
		G_P(\omega)=
		\begin{cases}
		F(\omega),&\text{ if } \omega \notin (\ul\omega_i,\ol\omega_i)\text{ for all }i,\\
		F(\ul \omega_i),&\text{ if } \omega \in (\ul \omega_i,m_i)\text{ for some }i,\\
		F(\ol \omega_i),&\text{ if } \omega \in [m_i,\ol \omega_i)\text{ for some }i.
		\end{cases}
\end{equation*}
Solving the monotone persuasion problem is thus equivalent to finding an \emph{optimal pooling set} $P$ that maximizes $\int_0^1 V(\omega)\df G_P(\omega)$. 

\begin{lemma}\label{l:T2}
	Each optimal pooling set $P$ takes one of the following forms:
	\begin{enumerate}
		\item  Interval disclosure $P=[0,\omega^*_L)\cup(\omega^*_R,1]$ with $m^*_L<\omega_L <\omega^*_L<\omega^*_R< \omega_R<m^*_R$.
		\item Cutoff rule $P=[0,\omega^*)\cup(\omega^*,1]$ with $m^*_L<\omega_L<\omega_R<m^*_R$.
		\item No disclosure $P=[0,1]$. 
	\end{enumerate}
\end{lemma}

\begin{proof}
We start with two simple claims.	

\begin{claim}[\citet{KMZ}] \label{l:KMZ}
Let $P$ be an optimal pooling set. 
\begin{enumerate}
	\item If $V$ is strictly concave on $[0,1]$, then $P=[0,1]$.
	\item If $V$ is s-shaped on $[0,1]$, then $P=(\omega^*,1]$, with $\omega^*<\omega_M<\E[\omega|\omega\in[\omega^*,1]]$.
\end{enumerate}
\end{claim}
\begin{proof}
Parts 1 and 2 follow from Corollary 1 and Theorem 1 in \citet{KMZ}.%
\end{proof}
\begin{claim}\label{l:cardinal}
Each optimal pooling set $P$ has the following properties.
\begin{enumerate}
	\item Each separating interval $[\ol \omega_i,\ul \omega_{i+1}]$, with $\ol \omega_i<\ul \omega_{i+1}$, is such that $[\ol\omega_i,\ul \omega_{i+1}]\subset [\omega_L,\omega_R]$. 
	\item There is at most one pooling interval $(\ul \omega_i,\ol\omega_i)$, with $\ul \omega_i<\ol\omega_i$, such that $m_i\in [0,\omega_L]$.
\end{enumerate} 
\end{claim}
\begin{proof}

To prove Part 1, suppose by contradiction that either $\ol \omega_i<\omega_L$ or $\ul \omega_{i+1}>\omega_R$. Then a pooling set that differs from $P$ only in that it pools all states in $(\ol \omega_i,\omega_L)$ or $(\omega_R,\ul\omega_{i+1})$ yields a strictly higher value by Part 1 of Claim \ref{l:KMZ}, as $V$ is strictly concave on $[0,\omega_L]$ and $[\omega_R,1]$. 

To prove Part 2, suppose by contradiction that there are two pooling intervals $(\ul \omega_i,\ol\omega_i)$ and $(\ul \omega_j,\ol\omega_j)$, with $\ol \omega_i\leq \ul \omega_j$, such that $m_i,m_j\in [0,\omega_L]$. Then a pooling set that differs from $P$ only in that it pools all states in $(\ul \omega_i,\ol \omega_j)$ yields a strictly higher value by Part 1 of Claim \ref{l:KMZ}, as $V$ is strictly concave on $[0,\omega_L]$ and the support of $G_{P}$ conditional on $(\ul \omega_i,\ol\omega_j)$ is a subset of $[0,\omega_L]$. %
\end{proof}

Suppose by contradiction that an optimal $P$ does not take Form 1, 2, or 3. If there is a separating interval $[\ol \omega_i,\ul \omega_{i+1}]$, with $\ol \omega_i<\ul \omega_{i+1}$, then $P=[0,\ol \omega_i)\cup(\ul \omega_{i+1},1]$, which is Form 1, leading to a contradiction. Indeed, by Part 1 of Claim \ref{l:cardinal}, we have $[\ol \omega_i, \ul \omega_{i+1}]\subset [\omega_L,\omega_R]$. By Part 2 of Claim \ref{l:KMZ}, we have $P\cap[\ol \omega_{i},1]=(\ul \omega_{i+1},1]$, as $V$ is s-shaped on $[\ol \omega_i,1]\subset [\omega_L,1]$. Analogously, $P\cap[0,\ul \omega_{i+1}]=[0,\ol \omega_{i})$. 

Next, suppose that there is no separating interval. Since $P$ does not take Forms 2 or 3, $P$ has two pooling intervals,  $(\ul \omega_i,\ol \omega_i)$ and $(\ul \omega_{i+1},\ol \omega_{i+1})$, with $\ol \omega_i=\ul \omega_{i+1}$, and $m_i\ge \omega_L$ or $m_{i+1}\le \omega_R$. W.l.o.g., suppose $m_{i+1}\le \omega_R$. By Part 2 of Claim \ref{l:cardinal}, $m_{i+1}>\omega_L$. By Part 2 of Claim \ref{l:KMZ}, $\ul\omega_{i+1}<\omega_L$. In summary, we have
\[
m_i<\ol\omega_{i}=\ul\omega_{i+1}<\omega_L<m_{i+1}\le\omega_R.
\]
Let $\omega=\ol\omega_{i}=\ul\omega_{i+1}$. Since $P$ is optimal, the marginal effect of changing $\omega$ should be 0, so, letting
\begin{align*}
\hat V(x)&=V(m_i)(F(x)-F(\ul \omega_i))+V(m_{i+1})(F(\ol \omega_{i+1})-F(x)), \quad\text{and}\\
\delta(x)&=V(m_i)+V'(m_i)(x-m_i)-V(m_{i+1})-V'(m_{i+1})(x -m_{i+1}),
\end{align*}
we have
\begin{gather*}
\hat V'(\omega)=V'(m_i)\frac{\df m_i}{\df \omega}(F(\omega)-F(\ul \omega_i))+V'(m_{i+1})\frac{\df m_{i+1}}{\df \omega}(F(\ol \omega_{i+1})-F(\omega))\\+(V(m_i)-V(m_{i+1}))f(\omega)\\
=[V(m_i)+V'(m_i)(\omega-m_i)-V(m_{i+1})-V'(m_{i+1})(\omega -m_{i+1})]f(\omega)=\delta(\omega)f(\omega)=0.
\end{gather*}
Thus, $\hat V'(\omega)=0$ iff $\delta(\omega)=0$.  Since $V$ is strictly concave on $[m_i,\omega_L]\subset[0,\omega_L]$ and strictly convex on $[\omega_L,m_{i+1}]\subset[\omega_L,\omega_R]$, we have $V(m_i)+V'(m_i)(\omega_L-m_i)>V(\omega_L)$ and $V(m_{i+1})+V'(m_{i+1})(\omega_L-m_{i+1})<V(\omega_L)$, respectively. Thus, 
\begin{gather*}
\delta(\omega_L)=V(m_i)+V'(m_i)(\omega_L-m_i)-V(m_{i+1})-V'(m_{i+1})(\omega_L-m_{i+1})>0.
\end{gather*}
Next, since $\delta(x)$ is linear in $x$ and $m_{i}<\omega<\omega_L$, we have
\beq\label{e:ppop}
\delta(m_{i})<\delta(\omega)=0<\delta(\omega_L).
\eeq
Since $\delta(m_{i})=(m_{i+1}-m_i)V'(m_{i+1})+V(m_i)-V(m_{i+1})$,
by \eqref{e:ppop}, we have
\beq\label{e:pop}
V'(m_{i+1})<\frac{V(m_{i+1})-V(m_i)}{m_{i+1}-m_i}.
\eeq
Let
\[
\nu(x)=V(x)-\frac{m_{i+1}-x}{m_{i+1}-m_i}V(m_i)-\frac{x-m_i}{m_{i+1}-m_i}V(m_{i+1}). 
\]
We have
\begin{gather*}
\nu'(x)=V'(x)-\frac{V(m_{i+1})-V(m_i)}{m_{i+1}-m_i}.
\end{gather*}

\begin{figure}[t!] 
	\centering 
	\includegraphics[scale=0.7]{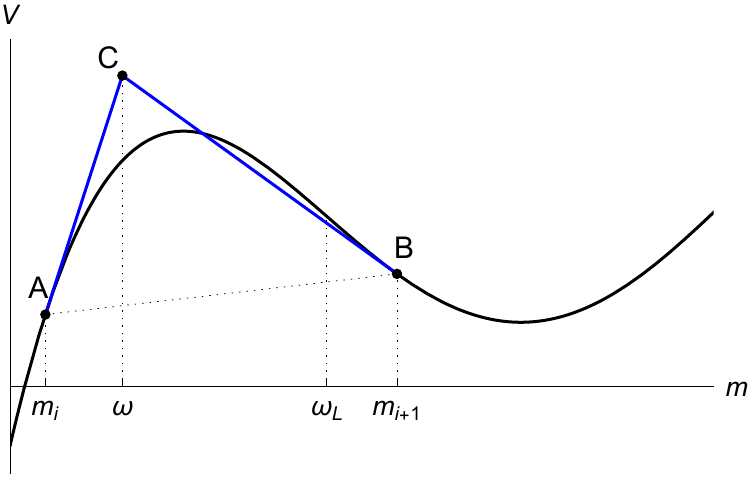}
	\bigskip
	\caption{Interval disclosure when $V$ is m-shaped.}
	\label{Fig:PropSwitchC}
	\legend{The tangents to $V$ at points $A$ and $B$ intersect at point $C$, because $\delta(\omega)=0$. Condition $\nu(x)>0$ for $x\in (m_i,m_{i+1})$ states that $V$ is above line segment $AB$.}
\end{figure}

Since $V$ is strictly concave on $[m_i,\omega_L]\subset[0,\omega_L]$ and strictly convex on $[\omega_L,m_{i+1}]\subset[\omega_L,\omega_R]$, $V'(x)$ is strictly quasiconvex on $[m_i,m_{i+1}]$, and so is $\nu'(x)$. Moreover, by \eqref{e:pop}, $\nu'(m_{i+1})<0$. Hence, $\nu'(x)$ is strictly single crossing from above on $[m_i,m_{i+1}]$. Thus, $\nu(x)$ is strictly quasiconcave on $[m_i,m_{i+1}]$. By $\nu(m_i)=\nu(m_{i+1})=0$, it follows that $\nu(x)>0$ for all $x\in(m_i,m_{i+1})$. Figure \ref{Fig:PropSwitchC} illustrates why $\nu(x)>0$.
Hence, a pooling set that differs from $P$ only in that it pools all states in $(\ul \omega_i,\ol \omega_{i+1})$ yields a strictly higher value, leading to a contradiction.
\end{proof}

By Proposition 3 in \citet{Kolotilin2017}, $P=[0,\omega^*_L)\cup(\omega^*_R,1]$, with $\omega_L <\omega^*_L<\omega^*_R< \omega_R$, is optimal iff \eqref{uml} and \eqref{umh} hold. Moreover, by Lemma \ref{l:T2}, $m^*_L<\omega_L$ and $m^*_R>\omega_R$. So, Part 1 of Theorem \ref{t-two} follows. If such $\omega^*_L<\omega^*_R$ do not exist, then $P$ takes Form 2 or 3 of Lemma \ref{l:T2}. Clearly, $P=[0,1]$ is suboptimal iff \eqref{u*e} holds for some $\omega^*\in(0,1)$. Moreover, if $P=[0,\omega^*)\cup(\omega^*,1]$ is optimal, then \eqref{u*m} holds, and $m_L^*<\omega_L$ and $m_R^*>\omega_R$ by Lemma \ref{l:T2}. So, Parts 2 and 3 of Theorem \ref{t-two} follow.
\qed

\subsection{Proof of Proposition \ref{L:M-C}}

Let $X\subset C$. For each $\omega\in [0,1]$, define  
\begin{gather*}
\ul c_X(\omega)=\sup\left(\{c\in C\backslash X: c\leq \omega\}\cup\{0\}\right),\quad 
\ol c_X(\omega)=\inf\left(\{c\in  C\backslash X: c> \omega\}\cup\{1\}\right),\\
\text{and}\quad \mu(\omega)=\E[\theta|\theta\in [\ul c_X(\omega),\ol c_X(\omega)]].
\end{gather*}
Observe that $\mu$ is a monotone signal such that $G_\mu=G_X$. Thus, $\mathcal G_C\subset \mathcal G_M$. 

Conversely, let $\mu$ be a monotone signal.  
For each $m$ such that $\mu(\omega)=m$ for some $\omega\in \supp(F)$, define
\begin{gather*}
\ul x_\mu(m)=\inf\left\{\omega\in \supp(F): \mu(\omega)=m\right\},\quad\ol x_\mu(m)=\sup\left\{\omega\in \supp(F): \mu(\omega)=m\right\},\\
\text{and}\quad X=\Big(\bigcup\nolimits_{m\in\mu(\supp(F))}(\ul x_\mu(m),\ol x_\mu(m)]\Big)\bigcap C.
\end{gather*}
Observe that $X$ is a censorship policy such that $G_X=G_\mu$. Thus, $\mathcal G_M\subset \mathcal G_C$. 
\qed

\end{appendix}

\end{document}